\documentclass[hideLIPIcs,nolineno,a4paper,UKenglish,cleveref,autoref,numberwithinsect]{socg-lipics-v2021}
\usepackage{graphicx}
\usepackage{color}
\usepackage{amsmath}
\usepackage{amssymb}
\usepackage{xspace}
\usepackage{epsfig}
\usepackage[dvipsnames]{xcolor}
\usepackage{xfrac}
\usepackage{soul} 
\usepackage{type1cm}

\usepackage[colorinlistoftodos,prependcaption,textsize=tiny]{todonotes}

\newcommand {\mm}[1] {\ifmmode{#1}\else{\mbox{\(#1\)}}\fi}

\newcommand {\scalprod}[2] {{\left\langle #1 , #2 \right\rangle}}

\newcommand{\ignore}[1]{}

\newsavebox{\smallProofsym}                 
\savebox{\smallProofsym}                               %
{
\begin{picture}(6,6)
\put(0,0){\framebox(6,6){}}
\put(0,2){\framebox(4,4){}}
\end{picture} 
}  

\makeatletter
\long\def\@makecaption#1#2{%
  \vskip\abovecaptionskip
  \sbox\@tempboxa{\small #1: #2}%
  \ifdim \wd\@tempboxa >\hsize
    \small #1: #2\par
  \else
    \global \@minipagefalse
    \hb@xt@\hsize{\hfil\box\@tempboxa\hfil}%
  \fi
  \vskip\belowcaptionskip}
\makeatother

\theoremstyle{plain}

\newtheorem*{supermaintheorem*}{Main Theorem}

\newtheorem*{supermaincorollary*}{Main Corollary}
\newtheorem*{supermaindefinition*}{Main Definition}

\newcommand{\Types}         {\mm{{\mathcal E}}}
\newcommand{\Dspace}        {\mm{{\mathbb D}}}
\newcommand{\Gspace}        {\mm{{\mathbb G}}}

\newcommand{\Rspace}        {\mm{{\mathbb R}}}
\newcommand{\Sspace}        {\mm{{\mathbb S}}}

\newcommand{\Qspread}         {\mm{{Q}_{\rm sprd}}}

\newcommand{\fspread}         {\mm{{f}_{\rm sprd}}}

\newcommand{\transpose}     {\mm{\tt T}}

\newcommand{\Sd}[1]         {\mm{{\rm Sd\,}{#1}}}
\newcommand{\trace}[1]      {\mm{{\rm tr\,}{#1}}}

\newcommand{\norm}[1]       {\mm{\|{#1}\|}}
\newcommand{\Edist}[2]      {\mm{\|{#1}-{#2}\|}}
\newcommand{\geodist}[2]    {\mm{{d_{\rm geo}}{\left({#1},{#2}\right)}}}

\newcommand{\hexdist}[2]    {\mm{{d_{\rm hex}}{\left({#1},{#2}\right)}}}

\newcommand{\diff}          {\mm{\,}{\rm d}}

\title{Quadratic Forms for Measuring Geometric Trees in 3-dimensional Space}
\titlerunning{Quadratic Forms for Measuring Geometric Trees in 3-dimensional Space}

\author{Yossi Bokor Bleile}{ISTA (Institute of Science and Technology Austria), Kloster\-neu\-burg, Austria and The University of Sydney, Sydney, Australia}{yossi.bokorbleile@ist.ac.at}{https://orcid.org/0000-0002-4861-9174}{}

\author{Emanuele Cortinovis}{ISTA (Institute of Science and Technology Austria), Kloster\-neu\-burg, Austria}{emanuele.cortinovis@ist.ac.at}{https://orcid.org/0009-0007-1331-776X}{}

\author{Herbert Edelsbrunner}{ISTA (Institute of Science and Technology Austria), Kloster\-neu\-burg, Austria}{herbert.edelsbrunner@ist.ac.at}{https://orcid.org/0000-0002-9823-6833}{}

\author{Shota Uka}{Technical University of Vienna, Vienna, Austria, and ISTA (Institute of Science and Technology Austria), Kloster\-neu\-burg, Austria}{e1211760@student.tuwien.ac.at}{https://orcid.org/0009-0004-3676-4325}{}

\authorrunning{Bokor Bleile, Cortinovis, Edelsbrunner, Uka}

\Copyright{Bokor Bleile, Cortinovis, Edelsbrunner, Uka}



\funding{This research was partially funded by the Austrian Science Fund (FWF) 10.55776/ESP9584724. For open access purposes, the author has applied a CC BY public copyright license to any author accepted manuscript version arising from this submission.}

\begin{document}
\maketitle

\begin{abstract}
  Tree-like structures appear in many areas of science, and their shapes can help understand the underlying processes they drive or that give rise to them.
  By thinking of these structures as geometric graphs in $\Rspace^3$, we gain access to tools from computational geometry and topology to study them.
  In this paper, we adopt the theory of quadratic forms to measure the directional spread of geometric graphs, and we introduce the hexplot model---equipped with a metric derived from the Fisher metric on the standard triangle---to visualize, measure, and collect statistics.
\end{abstract}

\section{Introduction}
\label{sec:1}

Exploring the morphology of tree-like structures is instrumental to understanding both their functional properties and the processes that generate them. 
Applications range from signal transmission in neurons to crystal growth and material porosity in physical systems. Quantifying and comparing such branching structures across different contexts remains a central challenge. 
Classical approaches rely on local geometric descriptors (such as branch length, curvature, bifurcation angles, tortuosity, branch diameter and taper, and Strahler order), summary profiles (such as the width function and Sholl analysis), or global scalar statistics (such as total path length and tree asymmetry indices); see e.g.\ \cite{Asc07,Ben20,Moh15,Wat24}.

\begin{figure}[hbt]
  \centering \vspace{0.0in}
    \includegraphics[width=\linewidth, keepaspectratio]{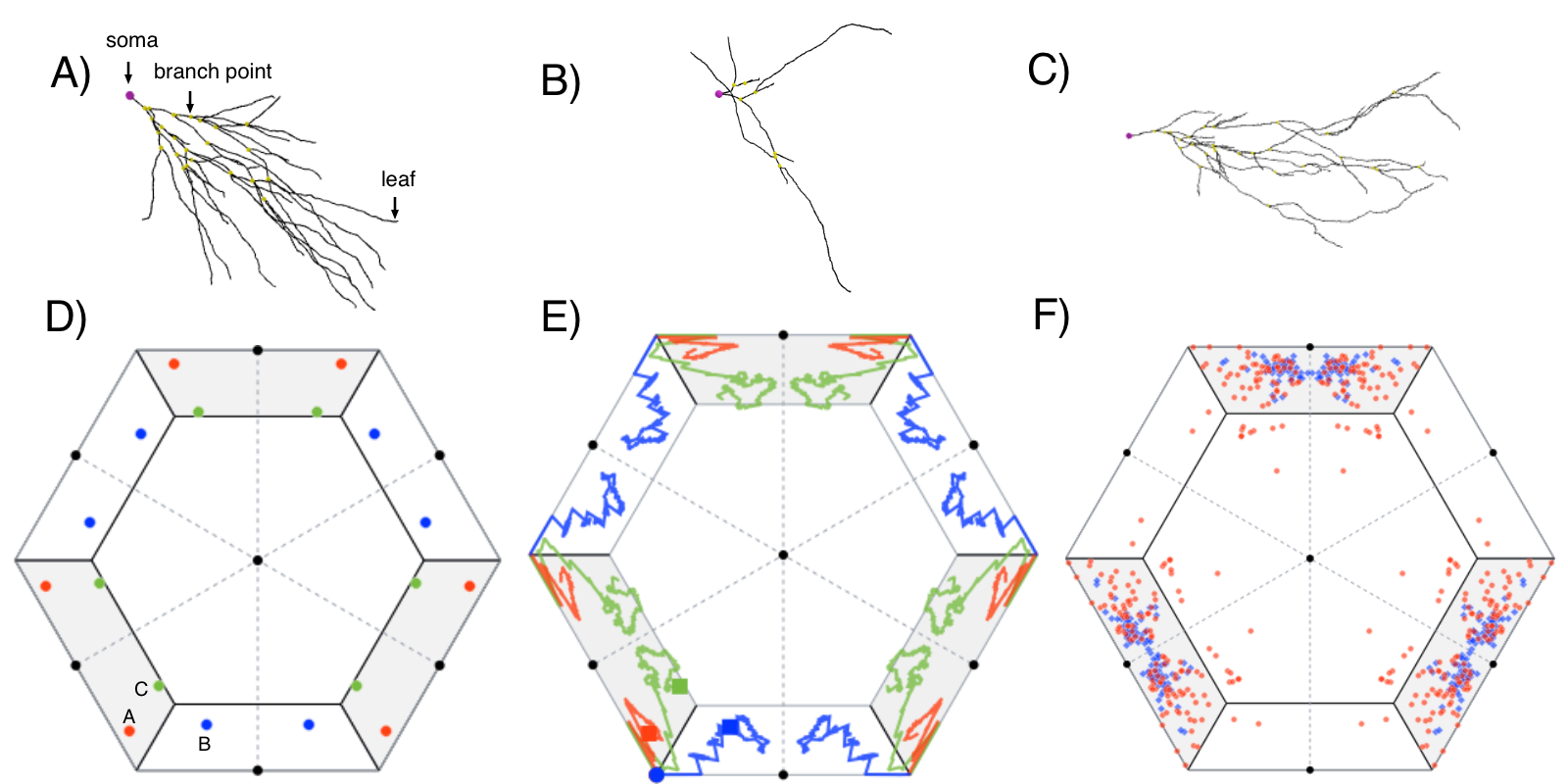}
  \vspace{-0.1in}
  \caption{\footnotesize A-C) Reconstructions of three apical dendrites (courtesy of Peter Jonas and Jake Watson at ISTA) from \emph{SWC} files.
  Each has the structure of a rooted geometric tree embedded in $\Rspace^3$.
  D) Each of the three dendrites corresponds to a $6$-tuple of points in the hexplot, of which the points in the lower left quadrangle are labeled.
  The hexplot has $6$-fold symmetry, with a quadrangular fundamental region indicated by dotted lines.
  E) Each dendrite is mapped to a continuous curve in the hexplot that reflects the evolution of the quadratic form as we consider the portion with progressively larger distance threshold from the soma.
  F) A comparison of two populations of dendrites in the hexplot.
  Points in the white and shaded rhombi represent dendrites with thin and elongated overall shape, respectively, and points in the central hexagon represent dendrites with overall round shape.}
  \label{fig:workflow}
\end{figure}
    
\smallskip
We focus on dendritic structures in $3$-dimensional space that can be viewed as geometric graphs embedded in $\Rspace^3$, thereby gaining access to geometric and topological tools that describe both local and global aspects of shape and branching. This representation enables principled comparisons based on the geometry, connectivity, and hierarchical organisation of the dendrites, and provides a natural interface for integrating local geometric measures (such as edge lengths, angles, and curvature) with topological invariants (such as graph homology and persistence summaries) that capture structural complexity across scales.

\smallskip
In this paper, we limit ourselves to quadratic forms whose level sets define ellipsoids; see e.g.\ \cite{HoJo13} for background on quadratic forms in general and \cite{Con97} for a delightful connection to number theory.
In particular, such quadratic forms allow us to quantify the \emph{directional spread} of a graph in $\Rspace^3$.
A quadratic form is given by a symmetric matrix, and the symmetric positive semi-definite $3$-by-$3$ matrices provide a natural representation of directional and spatial dispersion in $\Rspace^3$; for instance Schwartzman et al.\ \cite{SDT08} use the scatter matrix of axial data, whose eigenspectrum encodes both the mean orientation and the concentration of diffusion directions at each point in a spatial lattice. 
More generally, Schwartzman \cite{Sch16} develops log-normal distributions and geometric averaging operations on the cone of symmetric positive definite matrices, furnishing rigorous tools for statistical inference—including confidence regions and hypothesis tests—on populations of such quadratic forms, while respecting their positive-definiteness constraint.

\smallskip
To visualise the results, we introduce the \emph{hexplot model} for the family of quadratic forms defined by positive semi-definite $3$-by-$3$ matrices.\footnote{Using quadratic forms to measure directional spread can be generalised to higher dimensions, the hexplot model is however special to $3$-dimensional space.}
We equip this model with a metric structure obtained by conflating the Fisher metric on the standard triangle and its centrally reflected image.
Using this methodology, we construct unbiased descriptors of dendritic morphology that are invariant under rigid motion and scaling.
In addition, we use quadratic forms to recover the hierarchical organisation that distinguishes main from side branches.

\medskip
\noindent \textbf{Outline.}
Section~\ref{sec:2} introduces the mathematics of quadratic forms and explains how we use them to measure the directional spread of a geometric graph in $\Rspace^3$.
Section~\ref{sec:3} explains the hexplot model that visualizes the quadratic forms, including the metric defined on it, which quantifies the difference between two such forms.
Section~\ref{sec:4} uses the quadratic forms to decompose a tree into paths, and to visualize the directional spread as a function of the distance from its root.
Finally, Section~\ref{sec:5} concludes the paper by mentioning open questions and research directions.

\section{Squared Distance in 3 Dimensions}
\label{sec:2}

The average squared distance from a collection of planes in $\Rspace^3$ is given by a quadratic form whose level sets are ellipsoids; see \cite{HiCV52} for a general discussion of ellipses and ellipsoids.
We begin by introducing the relevant mathematical formalism and follow up by demonstrating how to use quadratic forms to measure the directional spread and length of a geometric graph in $\Rspace^3$.
The main motivating application is the study of neuronal morphologies.

\subsection{Mathematical Background}
\label{sec:2.1}

Consider an ellipsoid with axes of lengths $0 < 2a \leq 2b \leq 2c$ in $\Rspace^3$.
Assuming standard position, its axes align with the Cartesian coordinate axes in the same sequence, and the points of the ellipsoid have coordinates $x_1, x_2, x_3$ that satisfy
\begin{align}
  {x_1^2}/{a^2} + {x_2^2}/{b^2} + {x_3^2}/{c^2} &= 1.
  \label{eqn:ellipsoid}
\end{align}
The ellipsoid is called an \emph{oblate spheroid} if $a < b = c$, a \emph{prolate spheroid} if $a = b < c$, and a \emph{sphere} if $a = b = c$.
With some tolerance for the exact lengths, we will call the corresponding kinds of ellipsoids \emph{thin}, \emph{elongated}, and \emph{round}, in this sequence.
We also allow for degenerate cases, in which the extreme elongated ellipsoid is an \emph{elliptic cylinder} if $a \leq b < c = \infty$, and the extreme thin ellipsoid is a \emph{pair of planes} if $a < b = c = \infty$.

\smallskip
We call the left-hand side of \eqref{eqn:ellipsoid} a \emph{quadratic form} and the points for which it evaluates to $1$ the \emph{corresponding ellipsoid}.
Without insisting on the alignment of the axes but still centering the ellipsoid at the origin, we can write the quadratic form in matrix notation:
\begin{align}
  f(x) &= x^\transpose \cdot Q \cdot x
           =  \left[ ~x_1 ~x_2 ~x_3~ \right] \cdot
          \left[ \begin{array}{ccc}
                   A & D & F \\ D & B & E \\ F & E & C
          \end{array} \right] \cdot
          \left[ \begin{array}{c}
                   x_1 \\ x_2 \\ x_3
          \end{array} \right] 
    \label{eqn:matrix} \\
    &= A x_1^2 + B x_2^2 + C x_3^2 + 2D x_1 x_2 + 2E x_2 x_3 + 2F x_1 x_3 .
    \label{eqn:quadratic_form}
\end{align}
Without loss of generality, we may assume that the matrix is symmetric, and to specify an ellipsoid, it needs to be positive semi-definite; that is: $f(x) \geq 0$ for all $x \in \Rspace^3$.
We construct quadratic forms from planes that pass through the origin in $\Rspace^3$.
Letting $u \in \Sspace^2$ be the unit normal of such a plane, the squared distance of $x \in \Rspace^3$ from this plane is
\begin{align}
  \scalprod{x}{u}^2 
   &= \left( x^\transpose \cdot u \right) \cdot \left( u^\transpose \cdot x \right)
    = x^\transpose \cdot \left( u \cdot u^\transpose \right) \cdot x ,
\end{align}
in which the outer product, $u \cdot u^\transpose$, is a positive semi-definite matrix as in \eqref{eqn:matrix}.
Viewing the matrix as a linear transformation, there are generically three solutions to $Q \cdot x = \lambda x$, with $\lambda \in \Rspace$ and $x \neq 0$, referred to as the \emph{eigenvalues} and the corresponding unit \emph{eigenvectors} of the matrix; see standard texts in linear algebra, e.g.\ \cite{Str93}.
Because $Q$ is positive semi-definite, the eigenvalues are non-negative real numbers, denoted $\lambda_1, \lambda_2, \lambda_3$, with corresponding pairwise orthogonal eigenvectors, $e_1 , e_2 , e_3 \in \Sspace^2$.
The eigenvectors give the directions of the axes of the ellipsoid, and the eigenvalues specify their half-lengths, which are $1 / \sqrt{\lambda_1}$, $1 / \sqrt{\lambda_2}$, and $1 / \sqrt{\lambda_3}$.
There is a connection between the eigenvalues and the \emph{trace} of $Q$, which is the sum of diagonal entries, $A,B,C$, denoted $\trace{Q} = A+B+C$.
The trace defines a linear mapping; that is: $\trace{(Q+R)} = \trace{Q} + \trace{R}$ for any two $3$-by-$3$ square matrices $Q$ and $R$, and $\trace{(c \cdot Q)} = c \cdot \trace{Q}$ for any scalar $c$.
Importantly, the trace is equal to the sum of eigenvalues:
\begin{align}
  \trace{Q} &= \lambda_1 + \lambda_2 + \lambda_3 ;
\end{align}
see e.g.\ \cite{HoJo13}.
If $c \geq 0$ and $Q$ and $R$ are positive semi-definite, then so are $Q+R$ and $c \cdot Q$, and all their traces are non-negative.
The \emph{polar} of an ellipsoid, $E$, is the ellipsoid $E^*$ that bounds the points $y \in \Rspace^3$ satisfying $\scalprod{y}{x} \leq 1$ for every $x \in E$.
More specifically, if $x_0$ is a point of $E$, then the plane of points $y \in \Rspace^3$ with $\scalprod{x_0}{y} = 1$ touches $E^*$ at a single point, $y_0$, and the plane of points $x \in \Rspace^3$ that satisfy $\scalprod{y_0}{x} = 1$ touches $E$ in the point $x_0$.\footnote{This notion of polar body is heavily studied in the field of convex polytopes, but there seems to be a paucity in the literature on the special case of ellipsoids.}
\begin{figure}[hbt]
  \centering \vspace{-0.1in}
  \includegraphics[width=0.55\linewidth]{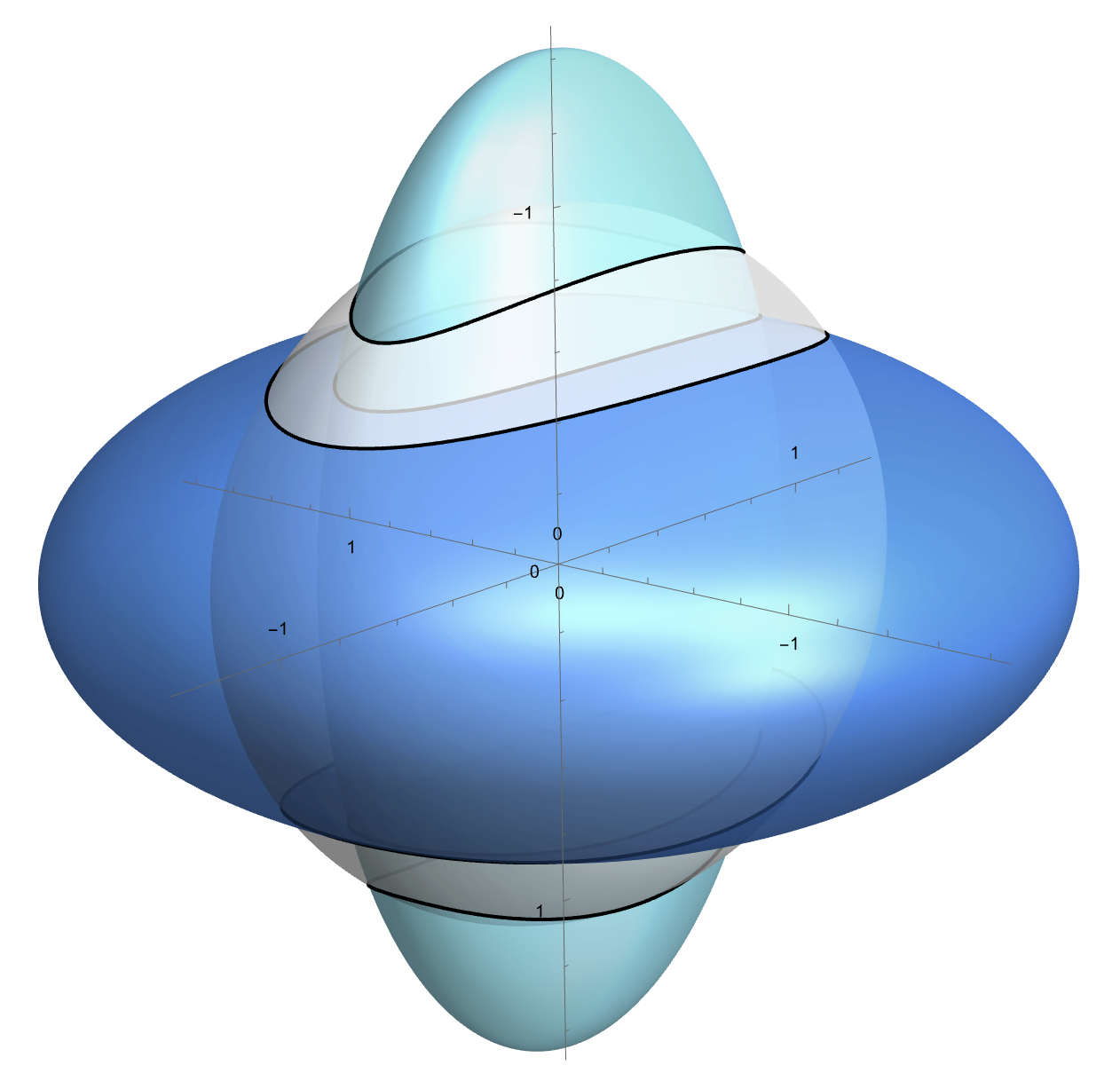}
  \vspace{-0.1in}
  \caption{\footnotesize A pair of polar ellipsoids, together with the unit sphere with respect to which polarity is defined.
  The axes of the \emph{blue} ellipsoid have half-lengths $1.8$, $1.3$, and $0.7$, while the axes of the polar ellipsoid have half-lengths $1/1.8$, $1/1.3$, and $1/0.7$, respectively.}
  \label{fig:polar}
\end{figure}
Indeed, the relation is symmetric for non-degenerate ellipsoids; that is: $E = (E^*)^*$.
The polar ellipsoid inherits the symmetries of the ellipsoid, so it is not difficult to see that the half-lengths of the axes of $E^*$ are the reciprocals of the half-lengths of the axes of $E$; see Figure~\ref{fig:polar} for an example.
This implies that the polar of a sphere is a sphere, and the oblate and prolate spheroids are polar to each other.

\smallskip
Throughout this paper, we will identify an ellipsoid with the quadratic form whose preimage of $1$ is the ellipsoid.

\subsection{The Eigenvalues Measure Directional Spread}
\label{sec:2.2}

We use the squared distance from planes to construct a quadratic form that measures the \emph{directional spread} of a collection of straight edges in $\Rspace^3$.
Letting $a_i, b_i \in \Rspace^3$ be the endpoints of the $i$-th edge in this collection, we set $w_i = \Edist{b_i}{a_i}$ and $u_i = (b_i - a_i) / w_i$.
 
The matrix of the quadratic form defined by these weights and vectors is
\begin{align}
  \Qspread &= \sum\nolimits_i w_i \left( u_i \cdot u_i^\transpose \right) ,
  \label{eqn:quadric_spread}
\end{align}
so $\fspread \colon \Rspace^3 \to \Rspace$ defined by 
\begin{align}
    \fspread (x) &= x^\transpose \cdot \Qspread \cdot x
                  = \sum\nolimits_i w_i x^\transpose \left( u_i \cdot u_i^\transpose \right)  x
\end{align} 
is the weighted sum of squared distances from the planes normal to the edges and passing through the origin of $\Rspace^3$.
The weight of each plane being the length of the corresponding edge implies that subdividing an edge into shorter edges does not affect the function.

\smallskip
The spread of directions is captured by the relations between the eigenvalues of $\Qspread$.
Specifically, the ellipsoid is round, elongated, of thin if $\Qspread$ has three roughly equal eigenvalues, one eigenvalue that is much smaller than the other two, or two eigenvalues that are much smaller than the third, respectively. 

Note that the polar of a round ellipsoid also tends to be round, while the polar of an elongated or thin ellipsoid tends to be thin and elongated, respectively.
A quantitatively more concrete expression of this characterization of shape can be formulated using the tools to be introduced in Section~\ref{sec:3}.

\subsection{The Trace Measures Length}
\label{sec:2.3}

Besides the directional spread, we can use the eigenvalues to measure length, namely by taking their sum.
Specifically, we have
\begin{theorem}
  \label{thm:trace_measures_length} 
  Let $a_i, b_i$ be the endpoints and $w_i = \Edist{b_i}{a_i}$ the length of the $i$-th edge in a collection in $\Rspace^3$, and $\fspread (x) = x^\transpose \cdot \Qspread \cdot x$ the quadratic form defined by weighting the squared distances from the normal planes with $w_i$.
  Then $\trace{\Qspread} = \sum_i \Edist{b_i}{a_i}$.
\end{theorem}
\begin{proof}
  Recall that $\Qspread$ includes one weighted outer product per edge, with trace equal to the length of the edge.
  It follows that the $i$-th edge contributes $w_i = \Edist{b_i}{a_i}$ to the trace of $\Qspread$.
  The claimed equations follows from the addiditivity of the trace.
\end{proof}

\section{The Hexplot}
\label{sec:3}

Thus far, the used quadratic forms are rotation and translation invariant.
To get a scale-independent representation of ellipsoids---which stresses the relation between the eigenvalues rather than their absolute sizes---we normalize, combine with the polar information, and draw them as sextuples of points in a regular hexagon.\footnote{Instead of ordering the eigenvalues, we have the symmetric group of degree (the six permutations of three elements) acting on the coordinates of every point.
Hence, the hexagon is a six-fold covering of the fundamental domain (a quadrangle) so that each triple of normalized eigenvalues maps to a sextuple of points in the hexagon, namely one point in each copy of the quadrangle.}
We show that this map is a homeomorphism to the hexagon, 
establish a metric on the hexagon, and generalize box plots to support the statistical analysis of collections of ellipsoids.

\subsection{Normalization and Antonelli Map}
\label{sec:3.1}

If we order the eigenvalues of an ellipsoid, we get a vector in the non-negative octant of $\Rspace^3$.
Normalizing this vector to unit length gives a point in the intersection of this octant and the unit sphere, denoted $\Sspace^2$.
Assuming $\lambda_1, \lambda_2, \lambda_3$ are the eigenvalues of $E$, then $1/\lambda_1, 1/\lambda_2, 1/\lambda_3$ are the eigenvalues of the polar, $E^*$,
and the normalized vectors are
\begin{align}
  \!\!\!\!\! \Lambda^\circ (E) &= \left\{ 
  \frac{\lambda_1}{\sqrt{\lambda_1^2+\lambda_2^2+\lambda_3^2}},
  \frac{\lambda_2}{\sqrt{\lambda_1^2+\lambda_2^2+\lambda_3^2}},
  \frac{\lambda_3}{\sqrt{\lambda_1^2+\lambda_2^2+\lambda_3^2}}
  \right\} ;
    \label{eqn:Lambdacirc} \\
  \!\!\!\!\! \Lambda^\circ (E^*) &= \left\{ 
  \frac{\lambda_2 \lambda_3}
       {\sqrt{\lambda_2^2 \lambda_3^2 + \lambda_1^2 \lambda_3^2 + \lambda_1^2 \lambda_2^2}},
  \frac{\lambda_1 \lambda_3}
       {\sqrt{\lambda_2^2 \lambda_3^2 + \lambda_1^2 \lambda_3^2 + \lambda_1^2 \lambda_2^2}},
  \frac{\lambda_1 \lambda_2}
       {\sqrt{\lambda_2^2 \lambda_3^2 + \lambda_1^2 \lambda_3^2 + \lambda_1^2 \lambda_2^2}} \right\} ,
    \label{eqn:Lambdacircstar}
\end{align}
in which we use set notation to indicate that each vector (of unordered components) corresponds to six conventional vectors whose components are ordered.

\smallskip
We note a side-effect of the normalization that will become important shortly.
Call $E$ \emph{non-degenerate} if all three eigenvalues are strictly larger than $0$.
For such an ellipsoid, the polar is unique, so both $\Lambda^\circ (E)$ and $\Lambda^\circ (E^*)$ are unique points in $\Sspace^2$.
This is no longer the case for a \emph{degenerate} ellipsoid, which has at least one vanishing eigenvalue.
Indeed, if $\lambda_1 = 0$ and $\lambda_2, \lambda_3$ are strictly positive, then $\Lambda^\circ (E^*) = \{1,0,0\}$ independent of the values of $\lambda_2$ and $\lambda_3$.
To cope with this ambiguity, we say two degenerate ellipsoids, $E$ and $E^*$, are \emph{polar} to each other if $E$ is the limit of a converging sequence of non-degenerate ellipsoids, $E_n$, such that $E^*$ is the limit of the sequence of non-degenerate polar ellipsoids, $E_n^*$.
Keep however in mind that different sequences converging to $E$ may give rise to different polar ellipsoids.

\smallskip
For visualization purposes, we further map the unit vectors to points in the \emph{standard triangle}, denoted $\Delta$, which is the intersection of the non-negative octant of $\Rspace^3$ with the plane of points that satisfy $x_1+x_2+x_3 = 1$.
Instead of further shortening the vectors, we prefer to use Antonelli's map \cite{Ant77}, which sends each coordinate of a unit vector to its square; that is:
\begin{align}
  \Lambda (E) &= \left\{ 
  \frac{\lambda_1^2}{\lambda_1^2+\lambda_2^2+\lambda_3^2},
  \frac{\lambda_2^2}{\lambda_1^2+\lambda_2^2+\lambda_3^2},
  \frac{\lambda_3^2}{\lambda_1^2+\lambda_2^2+\lambda_3^2}
  \right\};
    \label{eqn:Lambda} \\
  \Lambda (E^*) &= \left\{ 
  \frac{\lambda_2^2 \lambda_3^2}
       {\lambda_2^2 \lambda_3^2 + \lambda_1^2 \lambda_3^2 + \lambda_1^2 \lambda_2^2},
  \frac{\lambda_1^2 \lambda_3^2}
       {\lambda_2^2 \lambda_3^2 + \lambda_1^2 \lambda_3^2 + \lambda_1^2 \lambda_2^2},
  \frac{\lambda_1^2 \lambda_2^2}
       {\lambda_2^2 \lambda_3^2 + \lambda_1^2 \lambda_3^2 + \lambda_1^2 \lambda_2^2} \right\} ,
    \label{eqn:Lambdastar}
\end{align}
but note that the formula for $\Lambda (E^*)$ works only if at least two of the $\lambda_i$ are non-zero.
The main attraction of this map is it connects the geodesic distance on the sphere and the Fisher metric on the standard triangle.

\smallskip
Writing $\Lambda (E) = \{x_1,x_2,x_3\}$, it is customary to call $x_1, x_2, x_3$ the \emph{barycentric coordinates} of $\Lambda (E) \in \Delta$.
Keeping in mind that there are six permutations of the three coordinates, we note that $\Delta$ is also a $6$-fold covering of a smaller domain, namely a triangle in the barycentric subdivision of $\Delta$, denoted $\Sd{\Delta}$; see Figure~\ref{fig:hexplot}.
Indeed, the coordinates of any two points in the interior of a triangle in $\Sd{\Delta}$ have the same ordering.
Assuming for example the labeling of the vertices in Figure~\ref{fig:hexplot}, then every point $x_1 A + x_2 B + x_3 C$ in the interior of the south-west triangle, $\Delta_{\rm SW}$, satisfies $x_1 > x_2 > x_3$.
The reciprocals are ordered the opposite way, $1/x_1 < 1/x_2 < 1/x_3$, which implies that the corresponding point lies in the diagonally opposite triangle, which for $\Delta_{\rm SW}$ is the north-east triangle, $\Delta_{\rm NE}$.
\begin{figure}[hbt]
  \centering \vspace{0.1in}
  \includegraphics{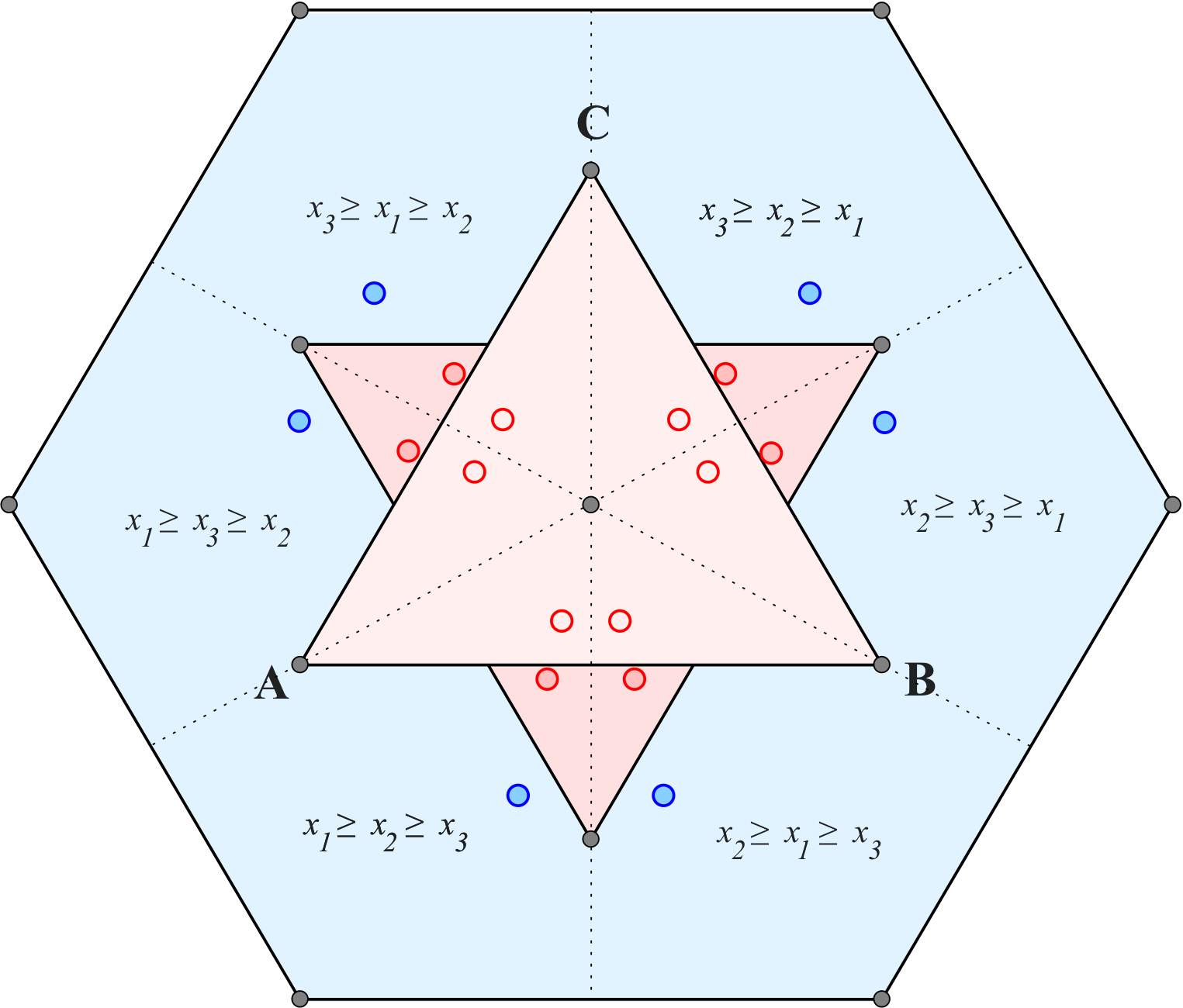}%
  \vspace{0.00in}
  \caption{\footnotesize The standard triangle, its central reflection, and the hexagon, each six-divided by the same three lines.
  An ellipsoid, $E$, maps to one point in each triangle of $\Sd{\Delta}$, $- E^*$ maps to one point in each triangle of $- \Sd{\Delta}$, and $(E,E^*)$ maps to one point in each quadrangle in the six-division of the hexagon.}
  \label{fig:hexplot}
\end{figure}

\smallskip
We combine the representation of an ellipsoid and its polar ellipsoid into one using the Minkowski sum of the standard triangle and its central reflection: $\Gspace = \Delta + (-\Delta)$,
which is a regular hexagon, as displayed in Figure~\ref{fig:hexplot}.
Specifically, the pair $(E, E^*)$ is sent to the point $\Gamma (E, E^*) = \Lambda (E) - \Lambda (E^*)$, or rather to a sextuple of points, one for each permutation of the eigenvalues of $E$.
Recall that $E$ is the polar ellipsoid of $E^*$, so $\Gamma (E^*, E) = - \Gamma (E, E^*)$ is well defined.
The barycentric subdivision of $\Delta$ extends to a division of $\Gspace$ into six quadrangles, each the Minkowski sum of a triangle in $\Sd{\Delta}$ with the central reflection of the opposite triangle in $\Sd{\Delta}$.
The pair $(E, E^*)$  thus maps to one point in the hexagon, which after unfolding is six points, one in each quadrangle, or three points if each is shared by two quadrangles, or one point if it is shared by all six quadrangles.
In the latter case, the point is $\Gamma (E, E^*) = \{0,0,0\}$, which implies $\Lambda(E) = \Lambda(E^*) = \{\frac{1}{3}, \frac{1}{3}, \frac{1}{3}\}$.
Indeed, every pair maps to a unique point in each hexagon, but this requires a proof.

\subsection{Bijectivity and Continuity}
\label{sec:3.2}

To prove properties of the map $\Gamma$ to the hexagon, we need to be precise about the class of ellipsoids that are mapped to the same point in $\Gspace$.
\begin{definition}
  \label{dfn:type}
  Two pairs of possibly degenerate ellipsoids, $(E,E^*)$ and $(D,D^*)$, have the same \emph{type} if $E$ and $D$ have the same set of three normalized eigenvalues, and so do $E^*$ and $D^*$; that is: $\Lambda(E) = \Lambda(D)$ and $\Lambda(E^*) = \Lambda(D^*)$. 
\end{definition}
We visualize the space of types as a $2$-dimensional submanifold with boundary in the shape of a hexagon in $\Rspace^6$, denoted $\Types$.
Each point in $\Types$ has six coordinates: the three barycentric coordinates of $\Lambda (E)$ and the three barycentric coordinates of $\Lambda (E^*)$.
The first three coordinates project $\Types$ to $\Delta$, and for almost all points of $\Types$, this projection is injective.
The only exception to injectivity are three sides of the hexagon, which project to the three vertices of $\Delta$.
Symmetrically, the second three coordinates project $\Types$ to $\Delta$, but now to the points of the polar ellipsoids, and injectivity is violated by the other three sides of the hexagon.
Indeed, these violations of injectivity are the main reason for defining $\Gamma \colon \Types \to \Gspace$ as a combination of two projections; that is: $\Gamma (E, E^*) = \Lambda (E) - \Lambda (E^*)$.
Importantly, this map is bijective, as we now prove.

\begin{lemma}[Bijectivity]
  \label{lem:bijectivity}
  The map $\Gamma \colon \Types \to \Gspace$ is bijective.
\end{lemma}
\begin{proof}
  We begin by proving the injectivity of $\Gamma$ while writing $x = \{x_1, x_2, x_3\}$ and $y = \{y_1, y_2, y_3\}$ for two points in $\Delta$, and $z = \{z_1, z_2, z_3\}$ for a point in $\Gspace$.
  Specifically, we show that for each $z \in \Gspace$, there is at most one $x \in \Delta$ such that $z = x-y$, in which $y = y(x)$ depends on $x$.
  For the six vertices of $\Gspace$, these are the six pairs of vertices of $\Delta$.
  We can therefore exclude them from the remainder of the argument and assume that at least two of the coordinates of $x$ are non-zero.
  Hence, 
  \begin{align}
    s &= \frac{x_1 x_2 x_3}{x_2 x_3 + x_1 x_3 + x_1 x_2}
    \label{eqn:normalize}
  \end{align}
  is well defined, and assuming $x = \Lambda (E)$, we have $y_i = s / x_i$ for the coordinates of $y = \Lambda (E^*)$.
  Multiplying with $x_i$ we get a quadratic equation, $x_i^2 - z_i x_i - s = 0$.
  It has either one or two solutions, and in the latter case, one solution is positive and the other negative.
  Since $x_i \geq 0$ we can discard the negative solution, so
  \begin{align}
    x_i &= \tfrac{1}{2} \left[ z_i + \sqrt{z_i^2 + 4s} \right] .
    \label{eqn:root2}
  \end{align}
  Fixing $z$ in $\Gspace$, we thus get unique points $x$ and $y = y(x)$ in $\Delta$ for each $s$.
  We now prove that there is at most one feasible choice for $s$.
  To this end consider the function that maps $s$ to the sum of right-hand sides in \eqref{eqn:root2} and its derivative:
  \begin{align}
    F_z (s) &= \tfrac{1}{2} \sum_{i=1}^3 \sqrt{ z_i^2 + 4s } ; 
      \label{eqn:Fz} \\
    \frac{\partial F_z}{\partial s} (s) &= \sum_{i=1}^3 \frac{1}{\sqrt{ z_i^2 + 4s }} ,
      \label{eqn:Fzprime}
  \end{align}
  in which we use $z_1 + z_2 + z_3 = 0$ to get \eqref{eqn:Fz}.
  Since $x_1 + x_2 + x_3 = 1$, we need $s$ such that $F_z (s) = 1$.
  Assuming $s > 0$, the derivative is well defined and positive, so $F_z$ is strictly increasing.
  Hence, if there is a solution it must be unique, which implies injectivity.

  \smallskip
  To see that $\Gamma$ is also surjective, we show that there is a solution for each $z \in \Gspace$.
  Observe that $\sum_i |z_i| = \sum_i |x_i-y_i| \leq \sum_i (x_i+y_i) = 2$, so the sum of the three absolute coordinates of $z$ is at most $2$.
  For $s = 0, 1$, we therefore have
  \begin{align}
    F_z (0) &= \tfrac{1}{2} \left[ |z_1| + |z_2| + |z_3| \right] \leq 1 ; \\
    F_z (1) &\geq \tfrac{1}{2} \left[ 2 + 2 + 2 \right] = 3 ,
  \end{align}
  and the intermediate value theorem implies the existence of an $s$ such that $F_z (s) = 1$.
\end{proof}

Before exploiting the bijectivity to get stronger properties of $\Gamma$, we present a technical lemma about the angle between the difference vectors defined by two ellipsoids, $E, D$, and their polars, $E^*, D^*$.
To formulate the claim, we let $x = \Lambda(E)$, $y = \Lambda(E^*)$ and $u = \Lambda(D)$, $v = \Lambda(D^*)$ be the corresponding points in $\Delta$, and assume their coordinates are sorted such that $x_1 \geq x_2 \geq x_3$ so $y_1 \leq y_2 \leq y_3$ and $u_1 \geq u_2 \geq u_3$ so $v_1 \leq v_2 \leq v_3$.
Setting
\begin{align}
  \frac{1}{s} &= \frac{1}{x_1} + \frac{1}{x_2} + \frac{1}{x_3}
  \mbox{\rm ~~and~~}
  \frac{1}{t} = \frac{1}{u_1} + \frac{1}{u_2} + \frac{1}{u_3} 
\end{align}
we see that the first relation agrees with \eqref{eqn:normalize}, and the coordinates of $y$ and $v$ are $y_i = s/x_i$ and $v_i = t/u_i$, for $i = 1,2,3$, respectively.
\begin{lemma}[Strong Cauchy--Schwarz Inequality]
  \label{lem:min_angle}
  Let $E, D$ be ellipsoids, $E^*, D^*$ their polars, and $x = \Lambda(E)$, $y = \Lambda(E^*)$, $u = \Lambda(D)$, $v = \Lambda(D^*)$ the corresponding points in $\Delta$.
  Then $\scalprod{u-x}{v-y} \leq \frac{1}{2} \Edist{u}{x} \Edist{v}{y}$, which for $u \neq x$ and $v \neq y$ is equivalent to having an angle of at least $60^\circ$ between the two vectors.
\end{lemma}

\begin{proof}
  The scalar product is the sum of three terms, each the product of two factors, which we call a \emph{matching pair}:
  \begin{align}
    \scalprod{u-x}{v-y} &= \sum_{i=1}^3 (u_i - x_i) \left( \frac{t}{u_i} - \frac{s}{x_i} \right) .
    \label{eqn:sp-one}
  \end{align}
  The respective first factors add to zero, $\sum_{i=1}^3 (u_i - x_i) = 0$, so two have the same sign and the third factor has a different sign.
  It is also possible that one or more of the three factors vanish, but we may think of this as a limiting case and assign signs arbitrarily.
  Accordingly, we call the factor with the unique sign \emph{long} and the two factors with the same sign \emph{short}.
  Similarly, the respective second factors add to zero, so we again have one long and two short factors.
  We claim that there are only two possible configurations of the six factors:
  \smallskip \begin{enumerate}
    \item[1.] none of the three matching pairs has two factors with the same sign; or
    \item[2.] exactly one of the three matching pairs has factors with the same sign, and these two factors are both short.
  \end{enumerate} \smallskip
  We prove this property and show that each possible configuration satisfies the claimed inequality.
  To begin consider a matching pair whose factors have the same sign.
  If $t \leq s$, then both factors are necessarily non-positive, and if $s \leq t$, then they are necessarily non-negative. 
  It follows that there cannot be two matching pairs of which one pair has two non-negative factors and the other has two non-positive factors.

  \smallskip
  Suppose first that the two long factors have different signs, which implies that also the first short and the second short factors have different signs.
  If we match the long factors with each other---and therefore the short factors with each other---we get three matching pairs, none of which has factors with the same sign.
  This is configuration 1, the scalar product is non-positive, so the angle is at least $90^\circ$ and therefore also at least $60^\circ$.
  On the other hand, if we match each long factor with a short factor, we get a contradiction because one matching pair has two non-negative factors while the other has two non-positive factors, which is impossible as argued above.

  \smallskip
  Suppose second that the two long factors have the same sign, which implies that all short factors have the same sign.
  If we match long with long and short with short, we get three matching pairs all with factors of the same sign, which is again impossible because this would mean one matching pair with two non-negative and another with two non-positive factors.
  The only remaining case is that we match each long factor with a short factor, which leaves one matching pair of two short factors.
  This is configuration 2.
  Letting $j$ be the corresponding index, we have
  \begin{align}
    \left( u_j - x_j \right)^2 
      &\leq \frac{1}{2} \left[ \sum_{i=1}^3 \left( u_i - x_i \right)^2 \right] 
    \mbox{\rm ~~and~~}
    \left( \frac{t}{u_j} - \frac{s}{x_j} \right)^2
      \leq \frac{1}{2} \left[ \sum_{i=1}^3 \left( \frac{t}{u_i} - \frac{s}{x_i} \right)^2 \right ] 
  \end{align}
  because both factors are short, which implies that in absolute value each is less than or equal to the corresponding absolute long factor.
  Hence, we have
  \begin{align}
    \!\!\!\!\left( u_j-x_j \right) \left( \frac{t}{u_j}-\frac{s}{x_j} \right)
      &\leq \sqrt{ \frac{1}{2} \sum_{i=1}^3 (u_i-x_i)^2 } \cdot \sqrt{ \frac{1}{2} \sum_{i=1}^3 \left( \frac{t}{u_i} - \frac{s}{x_i} \right) }
      = \frac{1}{2} \Edist{u}{x} \Edist{v}{y}.
  \end{align}
  But this is the only positive term in \eqref{eqn:sp-one}, so $\scalprod{u-x}{v-y} \leq \frac{1}{2} \Edist{u}{x} \Edist{v}{y}$, as claimed.
\end{proof}

\smallskip
We now go beyond bijectivity and prove that $\Gamma$ and $\Gamma^{-1}$ are both continuous.
To this end, we use the Euclidean metric in $\Rspace^3$ to measure distances in $\Delta$ and $\Gspace$, and the Euclidean metric in $\Rspace^6$ to measure distances in $\Types$.
For points $x, y, u, v \in \Delta$, we therefore write
\begin{align}
  d_\Types ((u, v), (x, y)) &= \sqrt{ \Edist{u}{x}^2 + \Edist{v}{y}^2 } .
\end{align}
Letting $z = x-y$ and $w = u-v$ be the corresponding points in $\Gspace$, we note that $\Edist{w}{z} \leq \Edist{u}{x} + \Edist{v}{y}$ by the triangle inequality, and $d_\Types ((u,v), (x,y)) \leq \Edist{u}{x} + \Edist{v}{y}$ because the square root function is concave.
\begin{theorem}[Bi-Lipschitz Homeomorphism]
  \label{thm:bi-Lipschitz_homeomorphism}
  The map $\Gamma \colon \Types \to \Gspace$ is a bi-Lipschitz homeomorphism, in which $\Gamma$ and $\Gamma^{-1}$ are both Lipschitz continuous with constant $2$.
\end{theorem}
\begin{proof}
  By Lemma~\ref{lem:bijectivity}, $\Gamma$ is bijective, so $\Gamma^{-1}$ exists.
  To see the Lipschitz continuity of $\Gamma$, consider ellipsoids $E, D$, and write $x = \Lambda(E)$, $y = \Lambda(E^*)$, $u = \Lambda(D)$, $v = \Lambda(D^*)$ as well as $z = x-y, w = u-v$.
  We have $\Edist{w}{z} \leq \Edist{u}{x} + \Edist{v}{y}$, and since the larger of the latter two distances is at least half of $\Edist{w}{z}$ and at most $d_\Types \left( (u,v), (x,y) \right)$, we get
  \begin{align}
    \Edist{w}{z} &\leq 2 d_\Types ((u,v), (x,y)) ,
    \label{eqn:homeo}
  \end{align}
  as needed.
  To see the Lipschitz continuity of $\Gamma^{-1}$, we consider the triangle with vertices $0$, $u-x$, $v-y$, and note that the edge opposite the vertex $0$ has length $\Edist{w}{z}$.
  By Lemma~\ref{lem:min_angle}, the angle at $0$ is at least $60^\circ$, which implies
  \begin{align}
    \Edist{u}{x} + \Edist{v}{y} &\leq 2 \Edist{w}{z} .
      \label{eqn:Lipschitz_two}
  \end{align}
  The left-hand side of \eqref{eqn:Lipschitz_two} is at least $d_\Types ((u,v), (x,y))$, which implies that $\Gamma^{-1}$ is Lipschitz continuous with constant $2$.
\end{proof}

\begin{remark*}
  The continuity of $\Gamma^{-1}$ can also be proven by appealing to a version of the Inverse Function Theorem, and thus without using Lemma~\ref{lem:min_angle}.
  In particular, \cite[Lemma~A.19]{Lee12} asserts that a continuous and bijective map from a compact space to a Hausdorff space is a homeomorphism.
  Indeed, the space of types is closed since it includes the limits of pairs $(x_n, y_n) = \left( \Lambda(E_n), \Lambda(E_n^*) \right)$, and it is bounded because every pair satisfies $\norm{x_n} \leq 1$ and $\norm{y_n} \leq 1$.
  This proof from general principles does however not imply that $\Gamma^{-1}$ is Lipschitz continuous.
\end{remark*}

\subsection{Hexagonal Box Plots}
\label{sec:3.5}
  
Letting $a_0, a_1, \ldots, a_n$ be an ordered sequence of real numbers, its \emph{box plot} is a graphical representation of five of these numbers: minimum, first quartile, median, third quartile, and maximum.
Assuming $n = 4k$, these are $a_0$, $a_{k}$, $a_{2k}$, $a_{3k}$, and $a_{4k}$.
The traditional way of drawing it is an axes aligned rectangle that stretches vertically from the first to the third quartile, has a horizontal bar at the median, and whiskers down to the minimum and up to the maximum; see the left panel of Figure~\ref{fig:boxplot}.
\begin{figure}[hbt]
  \centering
  \hspace{-0.0in}
  \includegraphics[width=\linewidth]{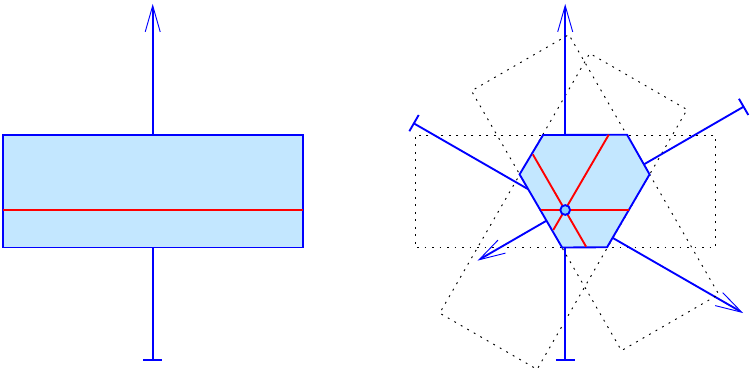}
  \caption{\footnotesize \emph{Left:} the conventional box plot of an empirical distribution.
  \emph{Right:} the hexagonal box plot, which combines three conventional box plots in one.}
  \label{fig:boxplot}
\end{figure}

For each ellipsoid, we get three eigenvalues and therefore three related box plots, which we combine by intersecting the boxes drawn on top of each other, centered at the shared midpoint of the three median lines, and rotated by $120^\circ$ relative to each other; see the right panel of Figure~\ref{fig:boxplot}.
We call this the \emph{hexagonal box plot}.
\begin{remark*}
  It is possible that the three strips defining a hexagonal box plot intersect in a convex polygon with fewer than six sides.
  However, in this case an extreme quarter of the population is the same for all three eigenvalues, so we may as well assume that the missing side has contracted to a corner of the plot.
  In other words, the corresponding line cannot be far from this corner, which justifies we still call it a \emph{hexagonal} box plot.
\end{remark*}

\subsection{The Metric}
\label{sec:3.6}

We use the Antonelli map once again, this time to define a metric on the hexagonal model of ellipsoid types by pulling back the geodesic distance on the unit sphere.

\begin{definition}
  \label{dfn:hexplot_Fisher_metric}
  The \emph{hexplot Fisher metric} is the map $d_{\rm hex} \colon \Gspace \times \Gspace \to \Rspace$ defined by
  \begin{align}
    \hexdist{z}{w} &= \geodist{\Lambda^\circ(E)}{\Lambda^\circ(D)} + \geodist{\Lambda^\circ(E^*)}{\Lambda^\circ(D^*)} ,
  \end{align}
  in which $z = \Gamma (E, E^*)$, $w = \Gamma (D, D^*)$, and $d_{\rm geo}$ denotes the geodesic distance between points on the unit sphere.
\end{definition}
To see that $d_{\rm hex}$ is indeed a metric, we note that $\hexdist{z}{w} = 0$ iff $z=w$, $\hexdist{z}{w} = \hexdist{w}{z}$ for all $z, w \in \Gspace$, and $\hexdist{z}{w} \leq \hexdist{z}{p} + \hexdist{p}{w}$ for a third point $p \in \Gspace$, simply because the triangle inequality holds for the geodesic distance on the sphere.

\begin{figure}[hbt]
  \centering \vspace{0.05in}
  \includegraphics[width=0.95\linewidth]{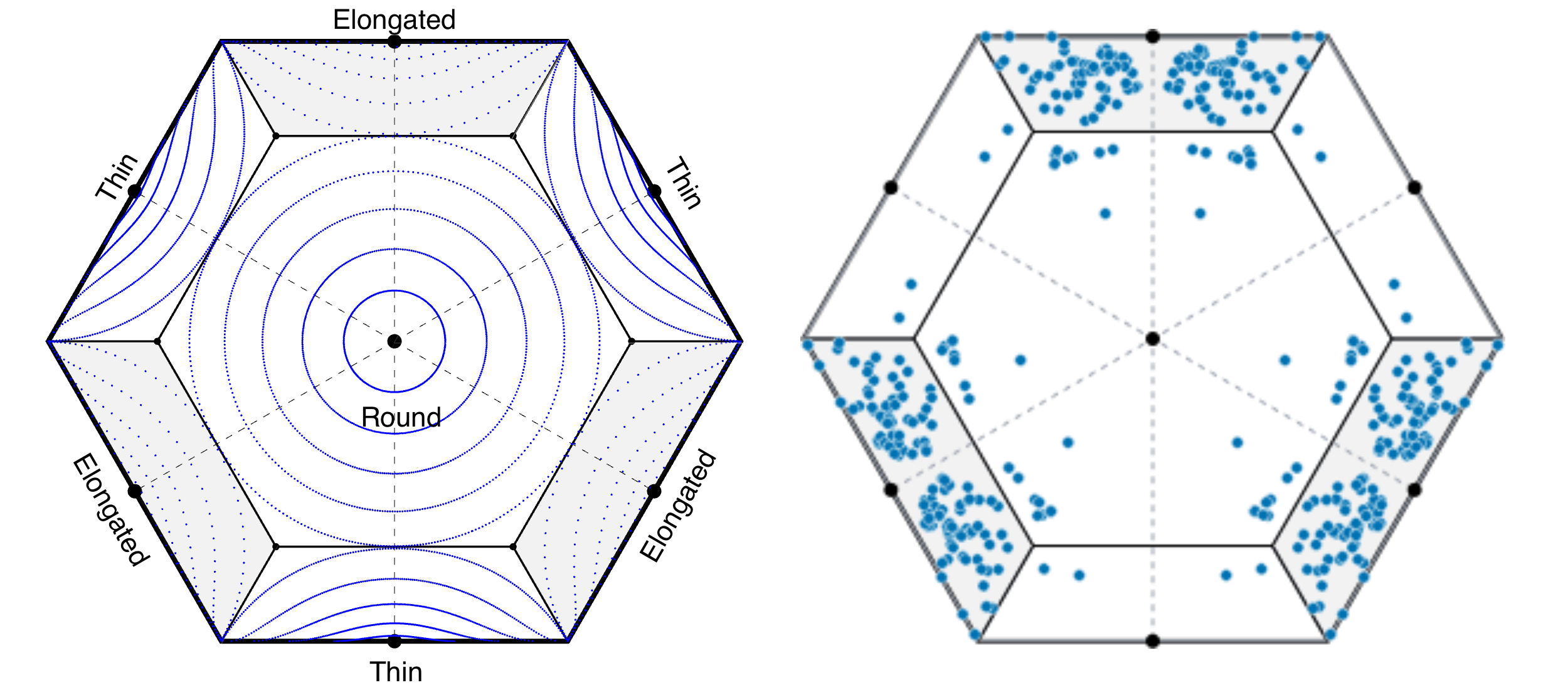}
  \vspace{-0.05in}
  \caption{\footnotesize \emph{Left:} an approximation of the Fisher--Voronoi tessellation of the origin and the midpoints of the six sides of the hexagon, together with five level sets each of the hexplot Fisher distance from these seven points.
  The \emph{shading} of the rhombi indicates that there are two kinds: the \emph{white} rhombi of thin or pancake-like polar ellipsoids and \emph{shaded} rhombi of elongated or cigar-like polar ellipsoids.
  \emph{Right:} an example population of dendritic structures, plotted using the \texttt{Quadrix} software \cite{quadrix}.
  The points range from elongated (inside the \emph{shaded} rhombi) to round (in the central region of the hexagon).
  The sparse population in the \emph{white} rhombi implies that thin polar ellipsoids are rare in this population.}
  \label{fig:Voronoi_and_levels}
\end{figure}

\begin{remark*}
  The pull-back of the geodesic distance on an orthant of the sphere to the standard simplex of the same dimension is traditionally known as the \emph{Fisher information metric} or the \emph{Fisher--Rao metric} between discrete probability distributions; see e.g.\ \cite{AmNa00}.
  It can alternatively be derived from the Shannon entropy by considering the Kullback--Leibler divergence \cite{KuLe51}, which measures the information loss if an encoding optimized for a different distribution is used.
  This divergence behaves like a distance but is not symmetric and also violates the triangle inequality.
  Taking the divergence to the infinitesimal limit and measuring a path by integrating the infinitesimal steps along it, we obtain the Fisher information metric from the shortest paths between their endpoints.
\end{remark*}

\smallskip
The geometry of the hexplot Fisher metric is not immediately obvious, so we visualize it by drawing the Voronoi tessellation of the center---the point $\{0,0,0\}$---and the midpoints of the six sides of the hexplot; which are really only two points: $\{\frac{1}{2}, \frac{1}{2}, -1\}$ for the white rhombi, and $\{-\frac{1}{2}, - \frac{1}{2}, 1\}$ for the shaded rhombi; see the left panel of Figure~\ref{fig:Voronoi_and_levels}.
The overall shape of a dendrite is shared by the polar ellipsoid of the pair, and if the representing point lies in the regions of $\{0,0,0\}$, $\{\frac{1}{2}, \frac{1}{2}, -1\}$, and $\{-\frac{1}{2}, - \frac{1}{2}, 1\}$, then we may call this shape round, thin (or pancake-like), and elongated (or cigar-like), respectively.
After unfolding, we have seven regions: the hexagon in the middle and six rhombi surrounding it.
These regions look like convex polygons, but at closer inspection it seems that the arcs that separate the hexagon from the six rhombi are not entirely straight.
The left panel of Figure~\ref{fig:Voronoi_and_levels} also show the level lines of the hexplot Fisher distance inside each region.
Here the level lines of the center are remarkably close to circles, and only for larger radii the subtle adaption to a more hexagonal shape becomes visible.

\section{Paths and Quadratic Forms}
\label{sec:4}

The quadratic forms can also be used to segment a dendritic structure into higher-level components than the edges, which includes the distinction between main trunk and side branches studied in this section.
Assuming information about the local thickness at every vertex, we could construct the main trunk by repeatedly moving to the child with the maximum radius, and compute the side branches recursively.
We consider another criterion that uses quadratic forms to make the choices.

\subsection{Path Decompositions}
\label{sec:4.1}

A \emph{path} is a linear sequence of edges in the tree, and it is \emph{maximal} if it starts and ends at a leaf each.
By a \emph{path decomposition} we mean a collection such that each edge belongs to exactly one path.
We do not allow two paths to pass through the same fork---even if the degree of that fork is four or larger---so we call a path decomposition \emph{minimal} if every fork of degree $k+1$ is interior to exactly one path and endpoint of $k-1$ paths.
We call a minimal decomposition \emph{hierarchical} if there is exactly one maximal path and every other path starts at a fork and ends at a leaf.
In such a decomposition, the maximal path is at the top of the hierarchy, all paths that start at interior vertices of the maximal path are its children, and the remaining paths are descendants further down the hierarchy.

\smallskip
To construct a hierarchical path decomposition, we assume a distinguished leaf, $r$, called the \emph{root}, and we direct each edge of the tree away from $r$.
If the tree is the representation of a neuronal dendrite, $r$ would be an artificially added node whose only child is the soma.
We thus have a tree in the computer science sense: other than the root, each node has a unique parent, and the remaining adjacent nodes are its children.
We construct the hierarchy top down and proceed greedily such that, given the starting point,
\smallskip \begin{description}
  \item[Criterion:] each selected path is to be as straight as possible.
\end{description} \smallskip
To quantify this criterion, let $f (x) = x^\transpose \cdot Q \cdot x$ be the quadratic form that measures the directional spread of a path, and write $(E, E^*)$ for the corresponding pair of ellipsoids.
This path is completely straight if $\Lambda(E) = \{ 0, 0, 1 \}$.
Setting $\Lambda(E^*) = \{\frac{1}{2}, \frac{1}{2}, 0 \}$, the corresponding point in the hexagon is $\Gamma_0 = \{ -\frac{1}{2}, -\frac{1}{2}, 1 \}$.
In the general case, we use the hexplot Fisher distance of $\Gamma (E, E^*)$ from $\Gamma_0$ to quantify how much the path deviates from being straight.

\smallskip
We need notation to describe the algorithm that computes the paths maximizing this criterion.
For a node $b$ of the tree, we write $k(b)+1$ for its degree, $b_0$ for the parent, and $b_1, b_2, \ldots, b_{k(b)}$ for the children.
For a given points $a$ and $b$, we write $f_{a,b} (x) = x^\transpose \cdot Q_{a,b} \cdot x$ for the quadratic form of the edge connecting the points, as defined in \eqref{eqn:quadric_spread}.
While traversing the tree, the algorithm computes the optimizing path starting at any edge in the tree, and stores the last point and the quadratic form of this path at this edge.
The algorithm is recursive, and we initially call it for the edge from the root, $r$, to its only child, $r_1$:
\begin{tabbing}
  m\=m\=mn\=mn\=mn\=mn\=mn\=mn\=mn\=mn\=\kill
  {\footnotesize 01} \> \> $\mbox{\tt function~} \mbox{\sc Pre-order} (a, b)$: initialize $\ell = b$ and $Q_\ell = Q_{a, b}$; \\*
  {\footnotesize 02} \> \> \> {\tt for} $i=1$ {\tt to} $k(b)$ {\tt do} $(\ell_i, Q_i) = \mbox{\sc Pre-order} (b, b_i)$; $Q = Q_{a,b} + Q_i$; \\*
  {\footnotesize 03} \> \> \> \> {\tt if} $\hexdist{\Gamma_0}{\Gamma(E, E^*)} < d_\ell$ {\tt then} $\ell = \ell_i$; $Q_\ell = Q$; $d_\ell = \hexdist{\Gamma_0}{\Gamma (E, E^*)}$ {\tt endif} \\*
  {\footnotesize 04} \> \> \> {\tt endfor};
    store $(\ell, Q_\ell)$ at the edge from $a$ to $b$ and
           {\tt return} $(\ell, Q_\ell)$.
\end{tabbing}
The algorithm gathers enough information so that the entire path decomposition can be extracted in a single additional traversal of the tree:
walk the path that starts at the root backward while labelling all its nodes, and recurse along the way to explore side paths at every fork.
Importantly, for each such side path, its last node has already been computed and stored at its first edge.
In summary, the decomposition that constructs paths greedily according to the straightness criterion runs in time linear in the number of nodes in the tree.

\smallskip
Other reasonable criteria for decomposing into paths are to maximize length, maximize centrality, or possibly combine different criteria into a weighted compromize.

\subsection{Evolving Quadratic Forms}
\label{sec:4.4}

To probe the shape of a dendrite near the beginning, near the end, and in between, we evaluate the quadratic form at every distance from its root.
Write $\Dspace$ for the tree in $\Rspace^3$, $r_0$ for its root, and $d_{\rm int} \colon \Dspace \to \Rspace$ for the map that sends each point $y \in \Dspace$ to its intrinsic distance from $r_0$.
\begin{figure}[hbt]
  \centering \vspace{-0.1in}
  \includegraphics[width=0.95\linewidth]{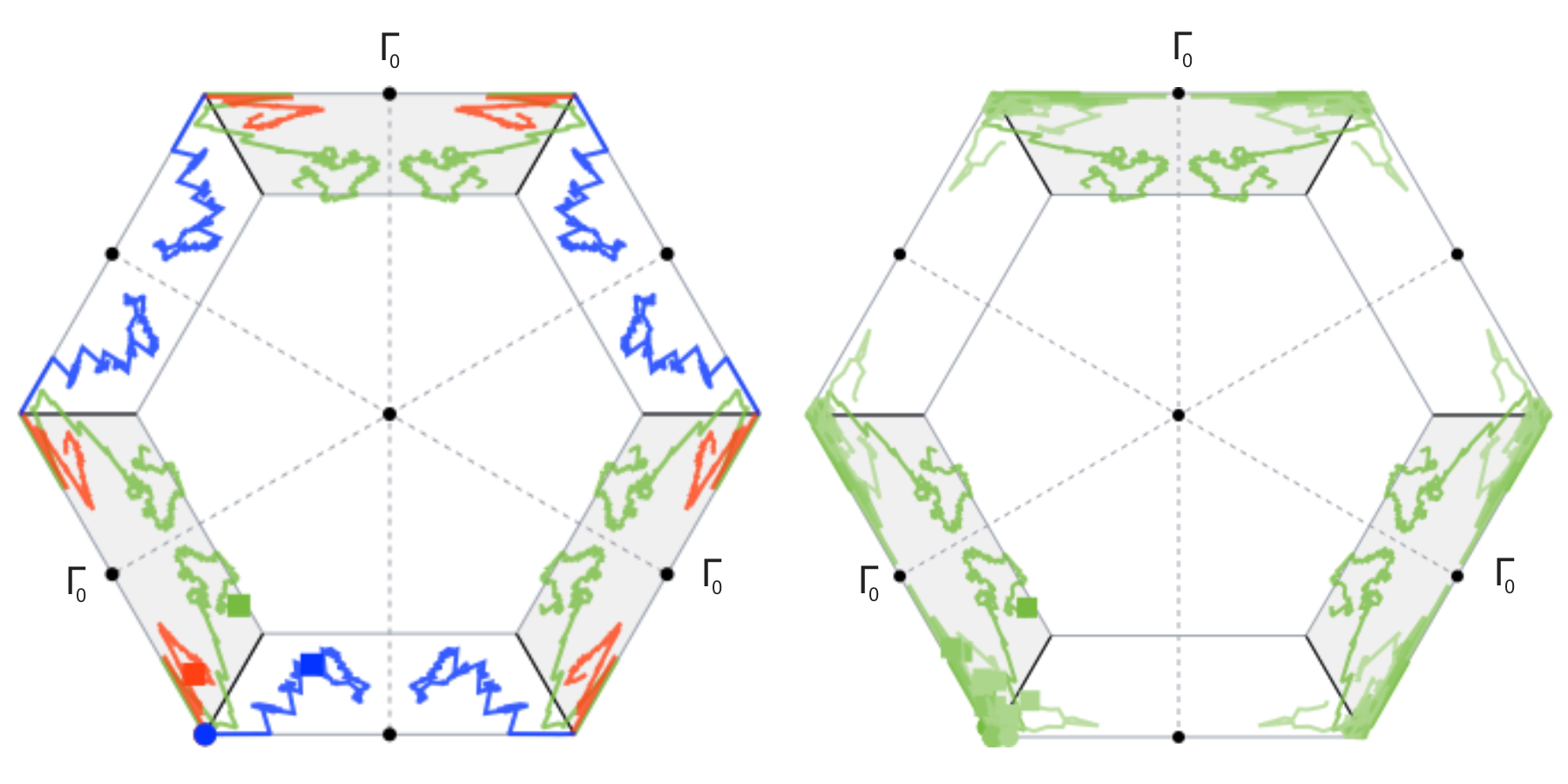}
  \vspace{-0.1in}
  \caption{\footnotesize Evolution of the directional spread of the dentrites in \cref{fig:workflow} (on the \emph{left}) and some of the paths in a decomposition of one of these dendrites (on the \emph{right}), for clarity we have stopped at depth $2$.}
  \label{fig:projected_trees}
\end{figure}

Hence, $\Dspace_t = d_{\rm int}^{-1} [0,t]$ are the points connected to $r_0$ by paths of length at most $t$, and we let $f_t (x) = x^\transpose \cdot Q_t \cdot x$ be the quadratic form that measures the directional spread of $\Dspace_t$.
By Theorem~\ref{thm:trace_measures_length}, the trace of $Q_t$ is the total length of this subtree.
The corresponding ellipsoid is $E_t = f_t^{-1} (1)$.
The polar ellipsoid, $E_t^*$, is unique if $Q_t$ has at least two non-zero eigenvalues, and we set $\Lambda (E_t^*) = \{ \frac{1}{2}, \frac{1}{2}, 0 \}$ if $\Lambda (E_t) = \{ 0, 0, 1 \}$.
Accordingly, we introduce
\begin{align}
  g_\Dspace \colon \Rspace_+ \to \Gspace \mbox{\rm ~~defined by~~}
  g_\Dspace (t) = \Gamma (E_t, E_t^*) 
\end{align}
for every $t \geq 0$.
It maps the tree to a curve that visualizes how the directional spread evolves as we move away from the root; see the left panel of Figure~\ref{fig:projected_trees}.\footnote{Since each pair of ellipsoids, $(E_t, E_t^*)$, maps to six points, we get six curves in $\Gspace$.
Any one of them will not necessarily be contained within the fundamental domain of $\Gspace$.
This is indeed the main reasons we use the hexagon---and not merely one of the quadrangles---since the former avoids sudden changes of direction when the curve is about to leave a quadrangle.}

\smallskip
It also makes sense to map individual paths in a decomposition of the tree to $\Gspace$; see the right panel of Figure~\ref{fig:projected_trees}.
Let for example $\gamma \subseteq \Dspace$ be the main path in the decomposition computed as described in Section~\ref{sec:4.1}; that is: the path that begins at $r_0$ and ends at another leaf of the tree.
The ellipsoid that corresponds to the first edge in $\gamma$ is thin, with $\Lambda (E) = \{0, 0, 1 \}$ and $\Lambda (E^*) = \{ \frac{1}{2}, \frac{1}{2}, 0 \}$.
The corresponding point in $\Gspace$ is $\Gamma_0 = \{-\frac{1}{2}, -\frac{1}{2}, 1\}$, which unfolds into three points, each shared by two quadrangles in the six-division of $\Gspace$.
The image of $g_\gamma \colon \Rspace_+ \to \Gspace$, is therefore a curve that starts at $\Gamma_0$ in $\Gspace$.

\subsection{Curves without Cusps}
\label{sec:4.5}

When we draw a curve within a single quadrangle of the six-division, its tangent vector gets reflected whenever the curve hits a side shared with a neighboring quadrangle.
The drawing is still continuous but has cusps needed to prevent entering other quadrangles.
This motivates us to draw the full hexagon and to prove that the image of $g_\gamma$ is differentiable if $\gamma$ is a generically smooth path in $\Rspace^3$.

\smallskip
We begin by adapting the quadratic form to the smooth case.
Let $\gamma \colon [0,L] \to \Rspace^3$ be an arc-length parametrization of a regular smooth path, i.e.\ the \emph{unit tangent}, $T(s)$, is defined for every $s \in [0,L]$.
We then introduce the quadratic form of such a path:
\begin{definition}
  \label{dfn:quadratic_form_in_limit}
  Let $\gamma \colon [0,L] \to \Rspace^3$ be the arc-length parametrization of a regular smooth path.
  Then
  \begin{align}
    \Qspread (\gamma) &= \int_{s=0}^{L}
      \left( T(s) \cdot T(s)^\transpose \right) \diff s 
      \label{eqn:smoothQspread}
  \end{align}
  and the corresponding quadratic form defined by $f_\gamma (x) = x^\transpose \cdot \Qspread (\gamma) \cdot x$ measures the directional spread of the path.
\end{definition}
To see that this is a sensible definition, consider a partition $0 = s_0 < s_1 < \ldots < s_n = L$ of $[0,L]$, write $\varphi_n \colon [0,L] \to \Rspace^3$ for the polygonal path with vertices $\gamma (s_i)$ for $0 \leq i \leq n$, and say a sequence of such polygonal paths \emph{rectifies} $\gamma$ if the length of the longest edge goes to zero and the length of the paths converges to the length of $\gamma$.
Importantly, the limit of the matrices whose corresponding quadratic forms measure the directional spread of the polygonal paths is the matrix defined in \eqref{eqn:smoothQspread}:
\begin{align}
  \lim\nolimits_{n \to \infty} \Qspread (\varphi_n) = \Qspread (\gamma) .
    \label{eqn:limit}
\end{align}

Assuming the curvature of $\gamma$ is non-zero at $s$, we can define the \emph{unit normal}, $N(s)$, and the \emph{unit binormal}, $B(s)$; see e.g.\ \cite{BrGi92}.
These two vectors, together with the unit tangent vector, $T(s)$, form an ortho-normal system classically referred to as the Frenet--Serre frame of $\gamma$ and $s$, and a \emph{singular point} is where this frame is not defined.
We call $\gamma$ \emph{tame} if the set of singular points is finite.
For each $0 \leq t \leq L$, write $\gamma_t$ for the portion of the path from $\gamma (0)$ to $\gamma (t)$, set $Q_t = \Qspread (\gamma_t)$, write $(E_t, E_t^*)$ for the corresponding pair of ellipsoids, and define $g_\gamma (t) = \Gamma (E_t, E_t^*)$.
Our goal is to prove that under the tameness assumption the image of $g_\gamma$ drawn in the hexplot has no cusps.
\begin{lemma}
  \label{lem:no_cusps}
  Let $\gamma \colon [0,L] \to \Rspace^3$ be the arc-length parametrization of a tame path in $\Rspace^3$.
  Then the image of $g_\gamma \colon (0,L) \to \Gspace$ is the trajectory of six continuously differentiable curves.
\end{lemma}
\begin{proof}
  By the fundamental theorem of calculus, $Q_t \colon [0,L] \to \Rspace^{3 \times 3}$ is a smooth family of symmetric bilinear forms.
  Moreover, its eigenvalues assemble in three continuously differentiable maps $\delta_i (t) \colon [0,L] \to \mathbb{R}$, for $i = 1,2,3$; see \cite[Chapter~II, \S~6.8]{Kat76}. 
  If $t>0$ then $Q_t \neq 0$, and hence $\delta_i(t) > 0$ for at least one $i$, so we can then normalize to obtain a curve in $\Delta$.
  Furthermore, by the tameness assumption at least two of the eigenvalues are non-zero.
  It follows that none of the associated ellipsoids maps to a vertex of $\Delta$, so the polar ellipsoid is uniquely defined for every $s \in (0, L)$.
  By composing with $\Gamma \colon \Delta \to \Gspace$, we obtain a continuously differentiable map, and taking into account the permutations of $\delta_1, \delta_2, \delta_3$ yields the image of the map $g_\gamma \colon (0,L) \to \Gspace$ and shows the claim.
\end{proof}

\section{Discussion}
\label{sec:5}

Motivated by the study of phenotypes of neuronal dendrites, we show how to use quadratic forms to measure the morphology of rooted trees in $3$-dimensional space, and how to visualize the results in what we call the \emph{hexplot of ellipsoid types}.
Indeed, we consider the introduction of the hexplot together with its Fisher metric as the main contribution of this paper.
The reported work suggests directions for further inquiry and raises yet open questions.
We list them in the order of increasing generality.
\smallskip \begin{itemize}
  \item Can we use quadratic forms to probe additional features of the dendritic phenotype, such as the chirality of paths, and the correlation between neighboring neurons?

  \item The hexplot reflects the relative size of eigenvalues and deliberately ignores absolute size and the direction of the eigenvectors.
  Are there elegant ways to add this information without substantially increasing the complexity of the visualization?

  \item How does the framework introduced in this paper extend to more general settings, such as quadratic forms defined by matrices that are not necessarily positive semi-definite, and to dimensions beyond three?
\end{itemize} \smallskip
A broad quest that relates to the second of the above items is the mathematical description of complex traits of neuronal dendrites, such as the uptake of information or other resources that may guide the way they explore the ambient space.
Can quadratic forms and the related tools developed in this paper contribute to this effort of rationalizing scientific data?


\end{document}